\documentclass[a4paper,UKenglish]{article}

\usepackage{microtype}

\usepackage{fullpage}
\usepackage{amsmath}
\usepackage{amssymb}
\usepackage{algorithm}
\usepackage{algorithmic}
\usepackage{authblk}
\usepackage{amsthm}
\usepackage{graphics}
\usepackage{graphicx}
\usepackage{hyperref}
\usepackage{wrapfig}
\usepackage{framed}
\usepackage{caption}
\usepackage{subcaption}

\newcommand{\Order}{\mathrm{O}}

\newcommand{\polylog}{\mathop{\mathrm{polylog}}\nolimits}
\newcommand{\poly}{\mathop{\mathrm{poly}}\nolimits}

\newcommand{\Exp}[1]{\mathbb{E}[#1]}

\renewcommand{\Pr}{\mathbb{P}}

\DeclareMathOperator*{\argmax}{arg\,max}

\theoremstyle{definition}

\newtheorem{theorem}{Theorem}

\newtheorem{lemma}{Lemma}

\date{} 
\usepackage{tikz}
\usetikzlibrary{positioning,chains,fit,shapes,calc,decorations.pathreplacing,graphs}

\usepackage{wrapfig}
\usepackage{framed}

\newcommand\clique{
    \begin{tikzpicture}
        \foreach \x in {1,...,8}{
            \pgfmathparse{(\x-1)*45}
            \node[draw,fill=black,circle,inner sep=1pt] (v\x) at (\pgfmathresult:10pt) {};
        }
        \foreach \x [count=\xi from 1] in {1,...,8}{
            \foreach \y in {\x,...,8}{
                \path (v\xi) edge[-] (v\y);
            }
        }
    \end{tikzpicture}
}

\newcommand\anticlique{
    \begin{tikzpicture}
        \foreach \x in {1,...,6}{
            \pgfmathparse{(\x-1)*60}
            \node[draw,fill=black,circle,inner sep=1pt] (v\x) at (\pgfmathresult:6pt) {};
        }
    \end{tikzpicture}
}

\title{Independent Set Size Approximation in Graph Streams \footnote{
    The work of GC is supported in part by
    European Research Council grant ERC-2014-CoG 647557,
    The Alan Turing Institute under EPSRC grant EP/N510129/1
    the Yahoo Faculty Research and Engagement Program and a Royal
    Society Wolfson Research Merit Award; JD is supported by a
    Microsoft Research Studentship; and CK by EPSRC grant  EP/N011163/1.}}

\author{Graham Cormode}
\author{Jacques Dark}
\author{Christian Konrad}
\affil{Department of Computer Science, Centre for Discrete Mathematics and its Applications (DIMAP), University of Warwick, Coventry, UK \\
  \texttt{\{g.cormode,j.dark,c.konrad\}@warwick.ac.uk}}

\begin{document}
\allowdisplaybreaks
\maketitle

\begin{abstract}
We study the problem of estimating the size of independent sets in
a graph $G$ defined by a stream of edges.
Our approach relies on the Caro-Wei bound, which expresses the desired
quantity in terms of a sum over nodes of the reciprocal of their
degrees, denoted by $\beta(G)$.
Our results show that $\beta(G)$ can be approximated accurately,
based on a provided lower bound on $\beta$.
Stronger results are possible when the edges are promised to arrive
grouped by an incident node.
In this setting, we obtain a value that is at most a logarithmic factor below
the true value of $\beta$ and no more than the true independent set size.
To justify the form of this bound, we also show an $\Omega(n/\beta)$ lower bound on any
algorithm that approximates $\beta$ up to a constant factor.
\end{abstract}

\section{Introduction}

For very large graphs, the model of streaming graph analysis, where
edges are observed one by one, is a useful lens.
Here, we assume that the graph of interest is too large to store in
full, but some representative summary is maintained incrementally.
We seek to understand how well different problems can be solved in
this model, in terms of the size of the summary, time taken to process
each edge and answer a query, and the accuracy of any approximation
obtained.
Variants arise in the model depending on whether edges can also be
removed as well as added, or if edges arrived grouped in some order,
and so on.

We study questions pertaining to independent sets within graphs.
Independent sets play a fundamental role in graph theory, and have
many applications in optimization and scheduling problems.
Given a graph, an independent set is a set of nodes such that there is no edge
between any pair.
There are conceptual links with matchings in graphs, since the dual
problem (find a set of edges such that no pair shares a node in
common) encodes the matching problem.
However, while matching permits efficient algorithms in the offline
setting, for independent set, the maximization problem is NP-hard, and
remains hard to approximate within $n^{1-\epsilon}$ for any
$\epsilon$.

The matching problem has received significant interest in the
streaming setting, and a large number of approximation algorithms are
known, with variations based on number of passes over the input data,
and whether the edges are weighted or unweighted.
Yet Independent Set is much less well understood.
In this paper, we provide algorithms and lower bounds that
characterize how well we can approximate the independent set problem
in the data stream model.
We focus on the cardinality version of the problem: the objective is
to output an estimate of the independent set size.
The size of the independent set can be linear in the number of nodes,
while we show that in some cases its cardinality can be estimated in
polylogarithmic space.

Our results rely on a combinatorial characterization of the
independent set size in terms of the degrees of nodes.
This reduces the focus to approximating a simple to describe function
(denoted as $\beta$),
yet this is still challenging in the streaming model.
Indeed, this $\beta$ function is hard to approximate when applied to an
arbitrary input sequence.
However we show that we can obtain good
approximations to $\beta$ when the input sequence corresponds to a graph.
We further distinguish the cases based on ordering in the input.
For arbitrary arrival orders, we obtain an approximation algorithm
for $\beta$, where the space reduces based on an assumed lower bound
on $\beta$ (which can be obtained based on average degree, for
example).
We then show that in the vertex arrival order, where each node arrives
along with all edges to nodes already in the graph, we can do much
better: the space cost drops to polylogarithmic, albeit for a
slightly different notion of approximation.
Our lower bounds characterize the minimum space cost necessary for any
algorithm which follows this approach, and help to
explain why some stronger approximation results are not possible in
either model.

We proceed as follows.  First, we provide necessary definitions and
notation; this allows us to state our results more formally (Section~\ref{sec:prelims}).
After surveying related work, we present our algorithms for the two
arrival models -- arbitrary edge arrivals (Section~\ref{sec:algs-edge}), and vertex-grouped edges
arrival (vertex arrival for short, Section~\ref{sec:algs-vertex}).
Last, we show space lower bounds for this problem based on a reduction
to a hard problem in communication complexity (Section~\ref{sec:lb}).

 \section{Definitions and Statement of Results}
\label{sec:prelims}
The Independent Set problem is most naturally modeled as a problem
over graphs $G = (V,E)$.
A set $U \subseteq V$ is an independent set if for all pairs $u, w \in
U$ we have $\{u, w\} \not\in E$, i.e. there is no edge between $u$ and
$w$.

We consider graphs defined by streams of edges.
That is, we observe a sequence of unordered pairs
$\{u, w\}$ which collectively define the (current) edge set $E$.
We do not require $V$ to be given explicitly, but take it to be defined
implicitly as the union of all nodes observed in the stream.
In the (arbitrary, possibly adversarial) edge arrival model, no
further constraints are placed on the order in which the edges
arrive.
In the vertex arrival model, there is a total ordering on the
vertices $\prec$ which is revealed incrementally.
Given the final graph $G$,
node $v$ ``arrives'' so that all edges
$\{u, v\} \in E$ such that $u \prec v$ are presented sequentially before the
next vertex arrives.
We do not assume that there is any further ordering among this group
of edges.

  Let $\alpha(G)$ be the {\em independence number} of graph $G$, i.e., the size of a maximum independent set in $G$.
 Let $\beta(G) = \sum_{v \in V(H)} \frac{1}{\deg_G(v) + 1}$ denote the Caro-Wei bound.
 It is well-known that $\alpha(G) \ge \beta(G)$, for every graph $G$~\cite{c79,w81}.
Our results focus on the problem of approximating $\beta(G)$ for
graphs presented as streams of edges.
We show the following three main results:

 \begin{enumerate}
 \item
In the Vertex Arrival Order model, we exhibit a one-pass randomized streaming algorithm that w.h.p. computes a value
  $\beta'$ such that $\beta' = \Omega(\beta(G) / \log n)$ and $\beta'
\le \alpha(G)$ using space $\Order(\log^3 n)$ bits.

\item A lower bound of $\Omega(\frac{n}{\beta(G) c^2})$ for
    computing a $c$-approximation to $\beta(G)$.
    The lower bound holds for the vertex arrival order and hence also
    in the (weaker) edge arrival order.

\item In the (adversarial) Edge Arrival Order, we present a one-pass randomized streaming algorithm that with high probability computes a
  $\phi$-approximation to $\beta(G)$ using space $\Order(\frac{n
  \polylog n}{\gamma \phi^2})$, where $\gamma$ is an arbitrary lower bound on $\beta(G)$.
  We also show a version of this algorithm which gives a $(1+\epsilon)$-approximation using
  $\Order(\frac{n \polylog n}{\gamma \poly(\epsilon)})$ space. Quantity $\beta(G)$ is
  bounded from below by $\frac{n}{\overline{d}+1}$, where $\overline{d}$ is the average
  degree of the input graph. Using $\gamma = \frac{n}{\overline{d}+1}$, the space of our algorithm
  becomes poly-logarithmic in $n$ for graphs of constant average degree, such as planar graphs or bounded
  arboricity graphs.

 \end{enumerate}

 \section{Related Work}
\label{sec:related}
There has been substantial interest in the topic of streaming
algorithms for graphs in the last two decades.
Indeed, the introduction of the streaming model focused on problems
over graphs~\cite{Henzinger:Raghavan:Rajagopalan:98}.
McGregor provides a survey that outlines key results on well-studied
problems such as finding sparsifiers, identifying connectivity
structure, and building spanning trees and matchings~\cite{McGregor:14}.
We expand on results related to matchings,
due to the similarity in problem statement.
For the unweighted case, the trivial greedy algorithm achieves a
maximal matching which is a 2-approximation to the size of the maximum
matching~\cite{Feigenbaum:Kannan:McGregor:Suri:Zhang:04}.
In the weighted case, a sequence of results have improved the
constant of approximation for this problem.
Most recently, a $(2+\epsilon)$ factor approximation was
presented~\cite{Paz:Schwartzman:17}.
In tandem with this line of work there has been a line of work
that seeks to approximate the cardinality of the maximum (unweighted) matching in
the stream.
This requires a distinct set of techniques.
Here, results have been recently shown by Assadi et
al.~\cite{Assadi:Khanna:Li:17}. For the dynamic version of the
problem (edges can be removed and inserted), they give an algorithm with a cost proportional
to $n^2/\alpha^4$, where $n$ denotes the number of vertices and
$\alpha$ the quality of the approximation, and
for the insert-only case, for arbitrary graphs, they give an algorithm
with space cost proportional to $n/\alpha^2$.
When the graph is sparse, characterized by having arboricity at most
$c$, a sampling-based algorithm can achieve an exponential improvement
in the space cost of $O(c \log^2 n)$ in order to
provide an approximation factor proportional to $c$~\cite{Cormode:Jowhari:Monemizadeh:Muthukrishnan:16}.
Our aim in this paper is to provide similar guarantees for estimating
the cardinality of independent sets.

Computing a maximum independent set is NP-hard on general
graphs \cite{k72} and is even hard to approximate within factor $n^{1-\epsilon}$, for any
$\epsilon > 0$ \cite{Hastad99,Zuckerman07}.
For this reason, often either special graph classes are considered
that admit reasonable approximations, or a different quality measure on the size of independent sets is used.
The Tur\'{a}n bound \cite{t48} implies that every graph has an independent set of size $n / (\overline{d}+1)$, where $\overline{d}$ is the average
degree of the input graph. Caro \cite{c79} and Wei \cite{w81}
improved this bound independently to $\beta(G)$.
The quantity $\beta(G)$ is an attractive bound on the size of a maximum independent
set since it is given by the degree sequence alone of a graph.
It is known that a simple greedy algorithm for maximum independent
sets computes an independent set of size at least $\beta(G)$ \cite{w81,griggs83}.
The algorithm iteratively picks a node of minimum degree, and removes
all neighbors from consideration --- note that this cannot be
simulated in the streaming model with small space.
There are parallels to other graph problems: for example, the minimum
vertex coloring problem is also NP-hard and hard to approximate within a factor of
$n^{1-\epsilon}$.
There is however a huge interest in computing $\Delta+1$-colorings, which is a quality bound also given by the degrees of the input graph.

It is known that the Caro-Wei bound gives polylogarithmic
approximation guarantees on graphs  which are of {\em polynomially
  bounded-independence}~\cite{hk15}, which means (informally) that the
size of the independent set in $r$-neighborhood around a node is
bounded in size by a polynomial in $r$.
This graph class includes unit interval and unit disc graphs.
The problem of finding independent sets themselves in the streaming model
has received some recent attention.
Halld\'{o}rsson et al. showed that an independent set of expected size $\beta(G)$
can be computed in the edge arrival model using $\Order(n \log n)$ space \cite{hhls16}.
The streaming independent set problem has been studied on interval graphs: In this model, the intervals arrive one-by-one.
The goal is to compute an independent set of intervals. There is an
algorithm that computes a $2$-approximation on general interval graphs
and a $1.5$-approximation on unit interval graphs which uses space
linear in the size of the computed independent set \cite{ehr16}.
Cabello and P{\'e}rez-Lantero gave polylogarithmic space streaming algorithms,
which approximate the size of independent sets of intervals \cite{Cabello2015}.


Our work is concerned with streaming approximations of the size of the maximum
independent set.
In \cite{hssw12}, very strong space lower bounds on approximating the size of a maximum independent set are given: Every $c$-approximation
  algorithm requires $\tilde{\Omega}(\frac{n^2}{c^2})$ space (which
  can also be achieved by sampling an induced subgraph and computing a maximum independent set in it using exponential time).
This strong lower bound provides a strong motivation for considering
related measures such as $\beta(G)$ instead.
Approximating $\beta(G)$ is essentially the same as approximating the
$-1$ negative frequency moment or harmonic mean of a frequency vector
derived from the graph stream.
This approach has been addressed via sampling approaches in the
property testing literature~\cite{grs11,ers16,abgpry16}, but has
received less attention from the perspective of streaming algorithms.
Braverman and Chestnut studied the problem of approximating the
negative frequency moments \cite{bc15} for general frequency vectors.
They consider only $(1+\epsilon)$-approximations and relate
  the space complexity to the stream length, i.e., the total weight of
  the input stream.
Our results evade their lower bounds, since the additional constraint
of being the degree distribution of a graph limits the shape of the
derived frequency vector, and precludes the pathological cases.

\section{Algorithm in the Edge-arrival Model}
\label{sec:algs-edge}
Suppose that we are given a bound $\gamma$ such that $\gamma < \beta(G)$.
We first give an algorithm with space
$\Order(\frac{n \log^3 n}{\epsilon ^4 \gamma})$ which approximates $\beta(G)$ within a factor of $1+\epsilon$ with high probability.
We then show how this algorithm can be turned into a $\phi$-approximation algorithm with space $\Order( \frac{n \log^3 n}{\gamma \phi^2})$.
Our algorithm and its analysis rely heavily on the use of degree
classes, which we introduce next.

\subsection{Degree Classes}

Let $c > 1$ be a real number used to define geometrically growing classes.
We partition the vertices of $G$ into classes $V_0 \cup V_1 \cup \dots V_{\lceil \log_c n \rceil -1}$
such that $v \in V_i$ iff $c^i \le \deg_G(v) < c^{i+1}$
(we assume that the input graph is connected and thus all vertex degrees are non-zero).
We define $\beta_i(G) := \sum_{v \in V_i} \frac{1}{\deg_G(v) + 1}$ which implies
$\beta(G) = \sum_{i} \beta_i(G)$. Furthermore, let $\beta'_i(G) := \frac{|V_i|}{ c^{i+1}+1}$, which implies:

\begin{eqnarray}
\beta'_i(G) \le  \beta_i(G) \le c \beta'_i(G). \label{eqn:199}
\end{eqnarray}

Let $g > 1$ be a parameter we set subsequently to control the
approximation factor.
Let $I_1$ be the set of class indices $i$ such that $\beta_i(G) \ge \frac{\beta(G) }{\lceil \log_c n \rceil g }$, and let $I_0$ be all other indices.
We call a degree class $i$ (or $V_i$) {\em heavy}, if $i \in I_1$, otherwise it is {\em light}.
We will argue that in order to obtain a good approximation to $\beta(G)$, it is enough to approximate $\beta'_i(G)$ for every heavy degree class $i$. We have
\[\sum_{i \in I_0} \beta_i(G) \le \sum_{i \in I_0} \frac{\beta(G) }{\lceil \log_c n \rceil g } \le \frac{\beta(G)}{g},\] which implies
$\sum_{i \in I_1} \beta_i(G) \ge \beta(G) ( 1 - \frac{1}{g})$. Furthermore, we obtain
\begin{eqnarray}
\sum_{i \in I_1} \beta_i(G) \le \beta(G) \le \frac{g}{g-1} \sum_{i \in I_1} \beta_i(G),   \label{eqn:900}
\end{eqnarray}
and by plugging Inequality~\ref{eqn:199} into Inequality~\ref{eqn:900}, we conclude
\begin{eqnarray}
 \sum_{i \in I_1} \beta'_i(G) \le \beta(G) \le \frac{g c}{g-1} \sum_{i \in I_1} \beta_i'(G) . \label{eqn:091}
\end{eqnarray}
Last observe that for $i \in I_1$, we have
$$
 \frac{\beta(G) }{\lceil \log_c n \rceil g } \le \beta_i(G) = \sum_{v \in V_i} \frac{1}{\deg_G(v) + 1} \le  \frac{|V_i|}{c^i + 1},
$$
which implies $|V_i| \ge \frac{\beta(G) }{\lceil \log_c n \rceil g }$, i.e., we establish a lower bound on the size of every heavy degree class.

\subsection{A $(1+\epsilon)$-approximation Algorithm}
Algorithm~\ref{alg:3} takes a uniform random sample $S$ of the vertices of the input graph and
maintains the degrees of these vertices while processing the stream. In a post-processing step,
vertices of $S$ are partitioned into degree classes $(S_i)_{i}$ in the same way
$V$ was partitioned in the previous subsection. Large enough sets $S_i$ then contribute
to our estimate for $\beta(G)$: By adjusting the parameters correctly, we ensure that heavy
degree classes $V_i$ give large samples $S_i$ with high probability, and we can accurately
estimate $|V_i|$ via $|S_i|$. Then, via Inequality~\ref{eqn:091}, this gives an accurate estimate of $\beta(G)$.

The main analysis of Algortihm~\ref{alg:3} is conducted in Lemma~\ref{lem:sampling-alg}, which gives approximation
and space bounds depending on parameters $\delta, g$ and $c$. In Theorem~\ref{thm:edge-arrival-gamma}, we optimize these parameters
so that space is minimized for obtaining a $(1+\epsilon)$-approximation. Last, in Theorem~\ref{thm:edge-arrival-phi}, we show how this
algorithm can be used to obtain a $\phi$-approximation, for an arbitrary value of $\phi$.
\begin{algorithm}[t]
 \begin{algorithmic}[1]
  \REQUIRE real value $\delta > 0$, real value $c > 1$, real value $g > 1$, $\gamma \le \beta(G)$
  \STATE $C \gets \frac{24}{\delta^2}$, $\displaystyle v_0 \gets \frac{\gamma}{\lceil \log_c n \rceil g}$, $\displaystyle p \gets \frac{C \log n}{v_0}$
  \STATE $S \gets $ subset of vertices obtained by sampling every vertex u.a.r. with probability $p$
  \WHILE{Processing the stream}
   \STATE For every $v \in S$: Compute degree $\deg_G(v)$
  \ENDWHILE
  \STATE Post-processing:
  \STATE $S_i \gets $ subset of $S$ of vertices $v$ with $c^i \le \deg_G(v) < c^{i+1}$
  \STATE $\beta' \gets 0$
  \FOR{$i = 0 \dots \lceil \frac{\log n}{\log c} \rceil -1$}
   \IF{$|S_i| \ge v_0 p / (1+\delta) $} \label{line:211}
    \STATE $\beta' \gets \beta' + \frac{|S_i|}{(c^{i+1}+1) p}$
   \ENDIF
  \ENDFOR
  \RETURN $\beta'$
 \end{algorithmic}
 \caption{Sampling based algorithm \label{alg:3} }
\end{algorithm}

\begin{lemma} \label{lem:sampling-alg}
  Let $\delta > 0$.
  If a value $\gamma \le \beta(G)$ is given to the algorithm, then
 Algorithm~\ref{alg:3} is a randomized one-pass streaming algorithm with space
 $\Order(\frac{n \log^3(n) g}{\gamma \delta^2 \log c})$ in the edge
 arrival model.
 With high probability it outputs a value $\beta'$ such that
 $$
  \frac{1}{1+ \delta} \sum_{i \in I_1} \beta'_i(G) \le \beta' \le (1+\delta) \beta(G) .
 $$
 If $\gamma > \beta(G)$, then the upper bound $\beta' \le (1+\delta) \beta(G)$ still holds w.h.p.
\end{lemma}

\begin{proof}
Suppose that $\gamma \le \beta(G)$. First, we prove that
for every $i$ with $|V_i| \ge v_0 = \frac{\gamma}{\lceil
    \log_c n\rceil g}$, the probability that the size of
set $S_i$ deviates from its expectation by more than a factor of
$1+\delta$ is small. To this end, suppose indeed that $|V_i| \ge v_0$.
Then, $\mu = \Exp{S_i} \ge v_0 p $, and
$$
\Pr \left[ \left| \left|S_i\right| - \mu\right| \ge \delta \mu \right] \le 2 \exp \left( - \frac{C \log(n) \delta^2}{2} \right) \le n^{-\frac{C \delta^2}{8}} \le n^{-2},
$$
for $C \ge \frac{16}{\delta^2}$, applying a standard Chernoff bound.
This proves that with high probability the condition in
Line~\ref{line:211} is fulfilled for every heavy degree class defined by the threshold $v_0$.

Next, suppose that $|V_i| \le v_0 / (1+\delta)^2$. Then, $\Exp{S_i}
\le v_0 p / (1+\delta)^2$, and by a similar Chernoff bound argument,
\begin{eqnarray*}
 \Pr \left[ |S_i| \ge \frac{v_0 p}{(1+\delta)} \right] & \le & \Pr \left[ |S_i| > \frac{v_0 p}{(1+\delta)^2} \cdot (1+\delta) \right] \le
 \exp \left( - \frac{\delta^2 C \log(n)}{(2+\delta) (1+\delta)^2} \right) \\
 & \le & n^{-\frac{\delta^2 C}{12}} \le n^{-2},
\end{eqnarray*}
for $C \ge \frac{24}{\delta^2}$. Thus, degree classes with fewer than
$v_0 / (1+\delta)^2$ vertices are not considered in Line~\ref{line:211} with
high probability.

Since w.h.p. degree classes with fewer than $v_0 / (1+\delta)^2$ nodes do not contribute to the output value $\beta'$ (i.e., the condition in Line~\ref{line:211}
evaluates to false), and for all degree classes $i$ with $|V_i| \ge v_0$, the size $|S_i|$ is concentrated around its mean within a factor of
$1+\delta$ w.h.p., the following lower bound on the output $\beta'$ holds w.h.p.:
\begin{eqnarray*}
 \beta' & \ge & \sum_{i \in I_1} \frac{|S_i|}{(c^{i+1} + 1)p} \ge \sum_{i \in I_1} \frac{\frac{p |V_i|}{1+\delta}}{(c^{i+1} + 1)p} =
 \sum_{i \in I_1} \frac{|V_i|}{(c^{i+1} + 1)(1+\delta)} = \frac{1}{1+\delta} \sum_{i \in I_1} \beta'_i.
\end{eqnarray*}
Furthermore, using the same argument as above, i.e., the fact that the sizes of all sets $S_i$ that contribute to $\beta'$ are concentrated
around their means within a factor of $1+\delta$, we obtain the following upper bound on the output $\beta'$:
\begin{eqnarray*}
 \beta' \le \sum_{i \, : \, |V_i| \ge v_0 / (1+\delta)^2} \frac{|S_i|}{(c^{i+1} + 1)p} \le \sum_{i \, : \, |V_i| \ge v_0 / (1+\delta)^2} \frac{p |V_i|(1+\delta)}{(c^{i+1} + 1)p} \le (1+\delta) \beta(G).
\end{eqnarray*}

This concludes the first part of the proof.

For the second part, to see that the upper bound $\beta' \le (1+\delta) \beta(G)$ still holds if $\gamma > \beta(G)$,
recall that the sizes of all sets $S_i$ that contribute to $\beta'$ are concentrated around its expected size
within a factor of $1+\delta$. Since the sampling probability becomes smaller as $\gamma$ increases,
fewer degree classes contribute to $\beta'$ and the upper bound thus equally holds.

Last, concerning space requirements of our algorithm, in expectation we sample $n \cdot p = \Order(\frac{n \log^2(n) g}{\gamma \delta^2 \log c} )$ nodes and
compute the degree for each node. Hence,
space $\Order(\frac{n \log^3(n) g}{\gamma \delta^2 \log c})$ bits are
sufficient.
Using a Chernoff bound, it can be seen that this also holds with high probability.
\end{proof}
We now use the previous lemma to establish our main theorem.
\begin{theorem} \label{thm:edge-arrival-gamma}
 Let $\gamma \le \beta(G)$. Then, there is a randomized one-pass approximation streaming algorithm in the edge arrival model with space
 $\Order(\frac{n \log^3(n) }{\gamma \epsilon^4})$ that approximates $\beta(G)$ within a
 factor of $1+\epsilon$, with high probability.
 If $\gamma > \beta(G)$, then the algorithm uses the same space and with high probability outputs a value $\beta'$ with $\beta' \le (1+\epsilon) \beta(G)$.
\end{theorem}
\begin{proof}
Suppose first that $\gamma \le \beta(G)$.
We run Algorithm~\ref{alg:3} using values for $\delta, c$ and $g$, which we determine later. By Lemma~\ref{lem:sampling-alg} the algorithm returns a
value $\beta'$ such that $\frac{1}{1+\delta} \sum_{i \in I_1} \beta'_i(G) \le \beta' \le (1+\delta) \beta(G)$. Using Inequality~\ref{eqn:091}, this gives
$$
 \frac{\beta'}{1+\delta} \le \beta(G) \le \frac{g c}{g - 1} (1+\delta) \beta' .
$$
Thus, we obtain a $(1+\epsilon)$-approximation, if $\frac{g c}{g - 1} (1+\delta) \le 1+\epsilon$. It can be verified that this is fulfilled
if we set $g = \frac{10}{\epsilon}$, $c = 1+\frac{\epsilon}{10}$ and $\delta = \frac{\epsilon}{10}$.
The space requirements thus are $\Order(\frac{n \log^3(n) g}{\gamma \delta^2} \log c) = \Order(\frac{n \log^3(n) }{\gamma \epsilon^3} \log (1+\epsilon)) = \Order(\frac{n \log^3(n) }{\gamma \epsilon^4} )$,
using the fact that $\log (1 + \epsilon) < \epsilon$, for any
$\epsilon<1$.

Last, if $\gamma > \beta(G)$, then $\beta'_i(G) \le \beta' \le (1+\delta) \beta(G)$ equally applies, by Lemma~\ref{lem:sampling-alg},
and the upper bound equally holds.
\end{proof}

Last, we turn the algorithm of the previous theorem into an algorithm with approximation factor $\phi$.

\begin{theorem}\label{thm:edge-arrival-phi}
 Let $\phi > 2$ and suppose that $\gamma' \le \beta(G)$ is a given lower bound on $\beta(G)$.
 There is a randomized one-pass approximation streaming algorithm in the edge arrival model with space
 $\Order(\frac{n \log^3(n) }{\gamma' \phi^2})$ that approximates $\beta(G)$ within a
 factor of $\phi$, with high probability.
\end{theorem}
\begin{proof}
We run the algorithm as stated in Theorem~\ref{thm:edge-arrival-gamma} with values
$\gamma = \gamma' \cdot \phi^2$ and $\epsilon = 1/4$. Let $\beta'$ be the output of the algorithm of Theorem~\ref{thm:edge-arrival-gamma}.
Then our algorithm returns the value $\beta'$ if $\beta' \ge \gamma / (1+\epsilon)$, and $\gamma' \phi$ otherwise.

First, suppose that $\beta(G) \ge \gamma$. By Theorem~\ref{thm:edge-arrival-gamma}, with high probability,
it holds $\beta(G) / (1+\epsilon) \le \beta' \le \beta(G) (1+\epsilon)$, and thus the output of our algorithm is
$\beta'$, which constitutes a $(1+\epsilon)$-approximation.

Next, suppose that $\beta(G) \le \gamma / 2$. By Theorem~\ref{thm:edge-arrival-gamma}, with high probability,
it holds $\beta' \le \beta(G) (1+\epsilon)$ and thus the output of our algorithm is $\gamma' \phi$. Since
$\beta(G) \le \gamma / 2$ (and larger than $\gamma'$), this constitutes a $\phi$-approximation.

Last, if $\gamma / 2 \le \beta(G) \le \gamma$, then both outputs $\beta(G)$ and $\gamma' \phi$
give $\phi$-approximations.
\end{proof}

\section{Algorithm in the Vertex-arrival Model}
\label{sec:algs-vertex}

Let $v_1, \dots, v_n$ be the order in which the vertices appear in the stream. Let $G_i = G[\{v_1, \dots, v_i\}]$ be
the subgraph induced by the first $i$ vertices.

Let $n_{d,i} := | \{ v \in V(G_i) \, : \, \deg_{G_i}(v) \le d \} |$ be the number of vertices of degree at most $d$ in $G_i$,
and let $n_d = \max_{i} n_{d,i}$. We first give an algorithm, \textsc{DegTest}$(d, \epsilon)$, which with high probability
returns a $(1+\epsilon)$-approximation of $n_d$ using $\Order(\frac{1}{\epsilon^2} \log^2 n)$ bits of space.

In the description of the algorithm, we suppose that we have a random function
\textsc{coin}: $\left[0, 1\right] \rightarrow \{\texttt{false}, \texttt{true}\}$ such that
\textsc{coin}($p$) = \texttt{true} with probability $p$ and
\textsc{coin}($p$) = \texttt{false} with probability $1-p$.
Furthermore,  the outputs of repeated invocations of \textsc{coin} are independent.

\begin{algorithm}
\begin{algorithmic}[1]
 \REQUIRE Degree bound $d$, $\epsilon$ for a $1+\epsilon$ approximation
  \STATE $p \gets 1$, $S \gets \varnothing$, $m \gets 0$, $\epsilon' \gets \epsilon / 2$, $c \gets \frac{28}{\epsilon'^2}$
  \WHILE[The current subgraph is $G_i$]{stream not empty}
   \STATE $v \gets $ next vertex in stream
   \STATE \textbf{if} \textsc{coin}($p$) \textbf{then} $S \gets S \cup \{ v\}$ \textbf{end if} \COMMENT{Sample vertex with probability $p$}
   \STATE Update degrees of vertices in $S$, i.e., ensure that for every $u \in S$ $\deg_{G_i}(u)$ is known
   \STATE Remove every vertex $u \in S$ from $S$ if $\deg_{G_i}(u) > d$
   \STATE \textbf{if} $p = 1$ \textbf{then} $m \gets \max \{ m, |S| \}$ \textbf{end if}
   \IF{$|S| = c \log(n)$} \label{line:c}
    \STATE $m \gets c \log(n) / p$
    \STATE Remove each element from $S$ with probability $\frac{1}{1+\epsilon'}$
    \STATE $p \gets p / (1 + \epsilon')$
   \ENDIF
  \ENDWHILE
  \RETURN $m$
 \end{algorithmic}
 \caption{Algorithm $\textsc{DegTest}(d, \epsilon)$  \label{alg:1}}
\end{algorithm}
Algorithm $\textsc{DegTest}(d, \epsilon)$ maintains a sample $S$ of at most $c \log n$ vertices. It ensures that all
vertices $v \in S$ have degree at most $d$ in the current graph $G_i$ (notice that $\deg_{G_i}(v) \le \deg_{G_j}(v)$, for every
$j \ge i$). Initially, $p=1$, and all vertices of degree at most $d$ are stored in $S$. Whenever $S$ reaches the limiting
size of $c \log n$, we downsample $S$ by removing every element of $S$ with probability $\frac{1}{1+\epsilon'}$ and update
$p \gets p / (1+\epsilon')$. This guarantees that throughout the algorithm $S$ constitutes a uniform random sample
(with sampling probability $p$) of all vertices of degree at most $d$ in $G_i$.

The algorithm outputs $m \gets c \log(n) / p$ as the estimate for $n_d$, where $p$ is the largest value of $p$ that occurs
during the course of the algorithm. It is updated whenever $S$ reaches the size $c \log n$, since
$S$ is large enough at this moment to be use as an accurate predictor for $n_{d,i}$, and hence also for $n_d$.
\begin{lemma}\label{lem:degtest}
 Let $0 < \epsilon \le 1$. $\textsc{DegTest}(d, \epsilon)$ (Algorithm~\ref{alg:1}) approximates $n_d$ within a factor $1+\epsilon$ with high probability, i.e.,
 $$\frac{n_d}{1+\epsilon} \le \textsc{DegTest}(d, \epsilon) \le (1+\epsilon) n_d \, ,$$
 and uses $\Order(\frac{1}{\epsilon^2} \log^2 n)$ bits of space.
\end{lemma}
\begin{proof}
First, suppose that $n_d < c \log n$. Then the algorithm never downsamples the set $S$ and computes $n_d$ exactly
(and makes no error). 

Assume now that $n_d \ge c \log n$. For $i \ge 0$, let $j_i$ be the smallest index $j$ such that
$n_{d,j} \ge c \log n (1+\epsilon')^i (1+\epsilon'/2)$. We say that the algorithm is in phase $i$, if $p = 1 / (1+ \epsilon')^i$.

First, for any $i$, we argue that in iteration $k \le j_i$, the algorithm is in a phase at most $i+1$ w.h.p. Let $E_{k,i}$ be the event
that the transition from phase $i+1$ to $i+2$ occurs in iteration $k \le j_i$, and let $E$ be the event
that at least one of the events $E_{k,i}$, for every $k$ and $i$, occurs.
For $E_{k,i}$ to happen, it is necessary that the algorithm is in phase $i+1$ in iteration $k$. Assume that this is the case.
Then, since $n_{d, k} \le n_{d, j_i}$, the expected size of $S$ in iteration $k$
is
$$\Exp{S} = \frac{n_{d,k}}{p} \le \frac{c \log (n) (1+\epsilon')^i (1+\epsilon'/2)}{(1+ \epsilon')^{i+1}} =  \frac{c \log (n) (1+\epsilon'/2)}{1+\epsilon'},$$
and thus, by a Chernoff bound,
\begin{eqnarray*}
\Pr \left[ |S| \ge c \log n \right] & \le & \exp \left(- \frac{ (\frac{1+\epsilon'}{1+\epsilon'/2})^2 }{2 + \frac{1+\epsilon'}{1+\epsilon'/2} } \cdot
\frac{c \log (n) (1+\epsilon'/2)}{1+\epsilon'} \right) =  \exp \left(- \frac{ \frac{1+\epsilon'}{1+\epsilon'/2}  c \log (n)}{2 + \frac{1+\epsilon'}{1+\epsilon'/2} } \right) \\
& = & \exp \left(- \frac{ (1+\epsilon')  c \log (n)}{3 + 2\epsilon' } \right) \le \exp \left(- \frac{ c \log (n)}{3} \right) \le n^{-3},
\end{eqnarray*}
for $c \ge 21$.
Thus, by the union bound, the probability that $E$ occurs is at most $n^{-2}$.

We assume from now on that $E$ does not occur.
Let $F_i$ be the event that at the end of iteration $j_i$, the algorithm is in phase $i+1$. We prove now by induction
that all $F_i$ occur with high probability. Consider first
$F_0$. Conditioned on $\neg E$, the algorithm is in phase $0$ or $1$ after
iteration $j_0$. We argue that with high probability, the algorithm is in phase $1$ after iteration $j_0$. Suppose
that the algorithm is in phase $0$ in the beginning of iteration $j_0$. Then,
$\Exp{S} = \frac{n_{d, j_0}}{p} = n_{d, j_0} = c \log n (1+\epsilon'/2)$.
Thus, by a Chernoff bound,
\begin{eqnarray*}
 \Pr \left[ |S| \le c \log n \right] \le \exp \left( - c \log n (1+\epsilon'/2) (\frac{\epsilon'}{2+\epsilon'})^2  \right) = \exp \left( - c \log n \frac{\epsilon'^2}{4+ 2\epsilon'} \right) \le n^{-2},
\end{eqnarray*}
for $c \ge \frac{28}{\epsilon'^2}$, and hence, if the algorithm was in phase $0$ at the beginning of iteration $j_0$, then, with high probability,
the transition to phase $1$ would occur.

Assume now that both $\neg E$ and $F_i$ hold. Then, the algorithm is in phase
$i+1$ or $i+2$ at the end of iteration $j_{i+1}$. Suppose we are in phase $i+1$ at the beginning of iteration $j_{i+1}$.
Then, $\Exp{S} = \frac{n_{j_0, d}}{p} = n_{j_0, d} = c \log n (1+\epsilon'/2)$, and by the same Chernoff bound as above,
the transition to phase $i+2$ would take place with high probability, which implies that $F_{i+1}$ holds.

Let $j_{\text{max}}$ be the largest $j$ such that $c \log n (1+\epsilon'/2) (1+\epsilon')^j \le n_d$. As proved above, when the algorithm
terminates, then the output $m$ is either $c \log n (1+\epsilon')^{j_{\text{max}}}$ or $c \log n (1+\epsilon')^{j_{\text{max}}+1}$ with high probability.
Suppose first that the output is $m = c \log n (1+\epsilon')^{j_{\text{max}}}$. Since $m (1+\epsilon'/2) (1+\epsilon') \ge n_d$,
the algorithm computes a $(1+\epsilon'/2) (1+\epsilon') \le (1+2\epsilon')$-approximation. Suppose now that the output is
$m = c \log n (1+\epsilon')^{j_{\text{max}}+1}$. Since $m (1+\epsilon'/2) / (1+\epsilon') \le n_d$, we equally obtain a $(1+2\epsilon')$-approximation.
Since $\epsilon = 2 \epsilon'$, the algorithm returns a $(1+\epsilon)$-approximation.

Concerning the space requirements of the algorithm, at most $c \log n$ vertex degrees are stored, which requires $\Order(\frac{1}{\epsilon^2} \log^2 n)$ bits of space.
\end{proof}

Next, we run multiple copies of \textsc{DegTest} in order to obtain our main algorithm, Algorithm~\ref{alg:2}.

\begin{algorithm}
\begin{algorithmic}
  \STATE \textbf{for} every $i \in \{0, 1, \dots, \lceil \log n \rceil \}$, run in parallel:
  \STATE \quad $\tilde{n}_{2^i} = \textsc{DegTest}(2^i, 1/2)$
  \STATE \textbf{end for}
  \RETURN $\displaystyle \max \left \{ \frac{\tilde{n}_{2^i}}{2(2^i + 1)} \, : \, i \in \{0, 1, \dots, \lceil \log n \rceil \} \right \}$
 \end{algorithmic}
 \caption{Algorithm in the Vertex-arrival Order \label{alg:2}}
\end{algorithm}

\begin{theorem}
 Let $\gamma$ be the output of Algorithm~\ref{alg:2}. Then, the following holds with high probability:
 \begin{enumerate}
  \item $\gamma = \Omega( \frac{\beta(G)}{\log n})$, and
  \item $\gamma \le \alpha(G)$.
 \end{enumerate}
 Furthermore, the algorithm uses space $\Order(\log^3 n)$ bits.
\end{theorem}
\begin{proof}
 For $0 \le i < \lceil \log(n) \rceil$, let $V_i \subseteq V$ be the subset of vertices with $\deg_G(v) \in \{2^i, 2^{i+1}-1\}$.
 Then,
 \begin{eqnarray}
 \nonumber \beta(G) & = & \sum_{v \in V} \frac{1}{\deg_G(v) + 1} = \sum_{i} \sum_{v \in V_i} \frac{1}{\deg_G(v) + 1} \le
 \sum_{i} \frac{|V_i|}{2^{i}+1}.
 \end{eqnarray} \\
Let $i_{\text{max}} := \argmax_i \frac{|V_i|}{2^{i}+1}$. Then, we further simplify the previous inequality as follows:
\begin{eqnarray}
 \beta(G) \le \dots \le \sum_{i} \frac{|V_i|}{2^{i}+1} & \le & \lceil \log(n) \rceil \cdot \frac{|V_{i_{\text{max}}}|}{2^{i_{\text{max}}}+1} \le
 \lceil \log(n) \rceil \cdot \frac{|V_{\le i_{\text{max}}}|}{2^{i_{\text{max}}}+1}. \label{eqn:493}
 \end{eqnarray}
where $V_{\le i} = \cup_{j \le i} V_j$.
Let $d_{\text{max}} = 2^{i_{\text{max}}}$.
Since $|V_{i_{\text{max}}}| \le n_{d_{\text{max}}}$
and $\tilde{n}_{d_{\text{max}}} = \textsc{DegTest}(d_{\text{max}}, 1/2)$ is a $1.5$-approximation to $n_{d_{\text{max}}}$,
we obtain $\gamma = \Omega(\frac{\beta(G)}{\log n})$, which proves Item 1.

Concerning Item 2, notice that for every $i$ and $d$, it holds
\begin{eqnarray*}
  \alpha(G) \ge \alpha(G_{i}) \ge \beta(G_i) = \sum_{v \in V(G_i)} \frac{1}{\deg_{G_i}(v) + 1} \ge \sum_{v \in V(G_i): \deg_{G_i}(v) \le d} \frac{1}{\deg_{G_i}(v) + 1} \ge \frac{n_{i,d}}{d + 1},
 \end{eqnarray*}
 and, in particular, the inequality holds for $n_{d_{\text{max}}} = n_{i_{\text{max}}, d_{\text{max}}}$. Since the algorithm returns a value bounded by
 $\frac{\tilde{n}_{d_{\text{max}}}}{2 \cdot (d_{\text{max}} + 1)}$, and $\tilde{n}_{d_{\text{max}}}$ constitutes a $1.5$-approximation of $n_{d_{\text{max}}}$, Item~2 follows.

Concerning the space requirements, the algorithm runs $\Order(\log n)$ copies of Algorithm~\ref{alg:1} which itself
requires $\Order(\log^2 n)$ bits of space.
\end{proof}

\section{Space Lower Bound}
\label{sec:lb}
Our lower bound follows from a reduction using a well-known hard problem
from communication complexity.
Let $\textsc{DISJ}_n$ refer to the \emph{two-party set disjointness problem} for inputs of size $n$.
In this problem we have two parties, Alice and Bob. Alice knows $X \subset [n]$, while Bob knows $Y \subset [n]$.
Alice and Bob must exchange messages until they both know whether $X
\cap Y = \emptyset$ or $X \cap Y \neq \emptyset$.

Using $R(\textsc{DISJ}_n)$ to refer to the randomised (bounded error probability) communication complexity of $\textsc{DISJ}_n$, the following theorem is known.

\begin{theorem}[Kalyanasundaram and Schintger \cite{ks92}]\label{thm:comm}
\[R(\textsc{DISJ}_n) \in \Omega(n)\]
\end{theorem}

To get our lower bound, we will show a reduction from randomised set disjointness to randomised $c$-approximation of $\beta(G)$.



\begin{theorem}\label{thm:main}
Every randomized constant error one-pass streaming algorithm that approximates $\beta(G)$ within a factor of $c$ uses space $\Omega(\frac{n}{\beta(G)c^2})$, even if
the input stream is in vertex arrival order.
\end{theorem}

\begin{proof}


Let $\textsc{Alg}_{c,n}$ be any streaming algorithm which takes as input a vertex arrival stream of an $n$-vertex graph $G$ and returns a $c$-approximation of $\beta(G)$ with probability $\frac{2}{3}$.

Suppose we are given an instance of $\textsc{DISJ}_{k}$. We will construct a
graph $G$ from $X$ and $Y$ which we can use to tell whether $X \cap Y = \emptyset$ by checking a
$c$-approximation of $\beta(G)$.

Let $z \ge 2$ be an arbitrary integer. Set $q = 2 z c^2$ and $a = kq$. Let $G = (V, E)$, where $V$ is partitioned into disjoint subsets
$A$, $B$, $C$, and $U_i$ for $i \in [k]$. These are of size $|A| = |B| = a$, $|C| = z$, and $|U_i| = q$. So $n := |V| = k q + 2a + z = 3kq + z = z(6kc^2 + 1)$.
Thus, $k \in \Theta(\frac{n}{z c^2})$ holds.

First consider the set of edges $E_0$ consisting of all $\lbrace u, v \rbrace$ with $u, v \in A \cup B$, $u \neq v$.
Setting $E = E_0$ makes $A \cup B$ a clique, while all other vertices remain isolated.

Figure \ref{fig:initial} shows this initial configuration. For clarity, we represent the structure using super-nodes and super-vertices.
A super-node is a subset of $V$ (in this case we use $A$, $B$, $C$, and each $U_i$). Between the super-nodes, we have super-edges representing the
existence of all possible edges between constituent vertices. So a super-edge between super-nodes $Z_1$ and $Z_2$ represents that
$\lbrace z_1, z_2 \rbrace \in E$ for every $z_1 \in Z_1$ and $z_2 \in Z_2$. The lack of a super-edge between $Z_1$ and $Z_2$ indicates that
none of these $\lbrace z_1, z_2 \rbrace$ are in $E$.

Now we add dependence on $X$ and $Y$. Let
$$E_X = \bigcup_{i \in [n] \setminus X} \left( \bigcup_{u \in U_i, v \in A} \lbrace\lbrace u, v\rbrace\rbrace \right) \text{ and } E_Y = \bigcup_{i \in [n] \setminus Y} \left( \bigcup_{u \in U_i, v \in B} \lbrace\lbrace u, v\rbrace\rbrace \right) \text{.}$$
So $E_X$ contains all edges from vertices in $U_i$ to vertices in $A$ exactly when index $i$ is not in the set $X$. Similarly for $E_Y$ with $B$, and $Y$.

Now let $E = E_0 \cup E_X \cup E_Y$. Adding these edge sets corresponds to
adding a super-edge to figure \ref{fig:initial} between $U_i$ and $A$
(or $B$) whenever $i$ is not in $X$ (or $Y$). Figures
\ref{fig:example1} and \ref{fig:example2} illustrate this.
In Figure~\ref{fig:example1}, the intersection is non-empty, which
creates a set of isolated nodes that push up the value of $\beta(G)$.
Meanwhile, there is no intersection in Figure~\ref{fig:example2}, so
the only isolated nodes are those in $C$.

Now, consider $\beta(G)$.
In the case where $X \cap Y = \emptyset$, we will have a super-edge connecting each $U_i$ to at least one of $A$ and $B$, so the degree
of each vertex in each $U_i$ is either $a$ or $2a$. 
Similarly, $A \cup B$ is a clique, so each vertex has degree at least $(2a-1)$. There are $2a$ such vertices, so they contribute
at most $\frac{2a}{(2a - 1) + 1} = 1$ to $\beta$. Vertices in $C$ are isolated and contribute exactly $z$ to $\beta$. Therefore, $z \leq \beta(G) \leq \frac{kq}{a} + 1 + z = z + 2$.

Now consider the case where $X \cap Y \neq \emptyset$. This means that there exists some $i \in X \cap Y$, and so $U_i$ will have no super-edges.
So each vertex in $U_i$ is isolated, and contributes exactly $1$ to $\beta$. There are $q$ such vertices, and also accounting for the contribution
of vertices $C$, we obtain $\beta(G) \geq q + z = z(2 c^2 + 1)$. 

Since the minimum possible ratio of the $\beta$-values between graphs in the two cases is at least $\frac{z(2 c^2+1)}{z+2} > c^2$ (using $z \geq 2$),
a $c$-approximation algorithm for $\beta(G)$ would allow us to distinguish between the two cases.


Now, return to our instance of $\textsc{DISJ}_k$. We can have Alice initialise an instance of $\textsc{ALG}_{c,n}$ and have all vertices in $A$, $C$, and each
$U_i$ arrive in any order. This only requires knowledge of $X$ because only edges in $E_0$ and $E_X$ are between these vertices and these are the only edges
that will be added so far in the vertex arrival model. Alice then communicates the state of $\textsc{ALG}_{c,n}$ to Bob. Bob can now have all vertices in $B$
arrive in any order. This only requires knowledge of $Y$ because only edges in $E_0$ and $E_Y$ are still to be added. Bob can then compute a
$c$-approximation of $\beta(G)$ with probability at least $\frac{2}{3}$, determining which case we are in and solving $\textsc{DISJ}_k$.

From Theorem~\ref{thm:comm}, we know that Alice and Bob must have communicated at least $\Omega(k)$ bits. However, all they communicated was the state of
$\textsc{ALG}_{c,n}$. Therefore,
$\Omega(k) = \Omega(\frac{n}{z c^2})$ bits was being used by $\textsc{ALG}_{c,n}$ at the time.

Consider again the graph $G$. The above argument shows that
in order to compute a $c$-approximation to $\beta(G)$, space $\Omega(\frac{n}{z c^2})$ is needed. Since $\beta(G) \geq z$ in both cases, we obtain the
space bound $\Omega(\frac{n}{\beta(G) c^2})$. Last, recall that $z$ and thus $\beta(G)$ can be chosen arbitrarily. The theorem hence holds for any value of $\beta(G)$.
\end{proof}

\begin{figure}
  \centering
    \resizebox{\textwidth}{!}{
    \begin{subfigure}[t]{0.325\textwidth}
      \begin{framed}
          \centering
          \definecolor{myred}{RGB}{217,185,185}
\definecolor{myblue}{RGB}{185,185,217}
\definecolor{mygreen}{RGB}{185,217,185}
\definecolor{mygrey}{RGB}{185,185,185}

\begin{tikzpicture}[thick,
  enode/.style={draw,circle,thick},
  unode/.style={enode,fill=myred},
  anode/.style={enode,fill=myblue},
  bnode/.style={enode,fill=mygreen},
  cnode/.style={enode,fill=mygrey},
]

\begin{scope}[start chain=going below,node distance=5pt]
    \foreach \i in {1,2}
        \node[unode,on chain] (u\i) [label=left: $U_{\i}$] {\anticlique};
    \node[on chain,minimum height=30pt] (dots) {$\vdots$};
    \foreach \i in {k-1,k}
        \node[unode,on chain] (u\i) [label=left: $U_{\i}$] {\anticlique};
\end{scope}

\begin{scope}[xshift=40pt,yshift=-20pt,start chain=going below,node distance=10pt]
  \node[anode,on chain] (a) [label=right: A] {\clique};
  \node[bnode,on chain] (b) [label=right: B] {\clique};
  \node[cnode,on chain] (c) [label=right: C] {\anticlique};
\end{scope}

\begin{scope}[every edge/.style={draw=black,very thick},
    every loop/.style={}]
    \path [-] (a) edge (b);
\end{scope}
\end{tikzpicture}
      \end{framed}
      \caption{Initial configuration.\label{fig:initial}}
    \end{subfigure}
    \begin{subfigure}[t]{0.325\textwidth}
      \begin{framed}
          \centering
          \definecolor{myred}{RGB}{217,185,185}
\definecolor{myblue}{RGB}{185,185,217}
\definecolor{mygreen}{RGB}{185,217,185}
\definecolor{mygrey}{RGB}{185,185,185}

\begin{tikzpicture}[thick,
  enode/.style={draw,circle,thick},
  unode/.style={enode,fill=myred},
  anode/.style={enode,fill=myblue},
  bnode/.style={enode,fill=mygreen},
  cnode/.style={enode,fill=mygrey},
]

\begin{scope}[start chain=going below,node distance=5pt]
    \foreach \i in {1,2,...,5}
        \node[unode,on chain] (u\i) [label=left: $U_{\i}$] {\anticlique};
\end{scope}

\begin{scope}[xshift=50pt,yshift=-20pt,start chain=going below,node distance=10pt]
  \node[anode,on chain] (a) [label=right: A] {\clique};
  \node[bnode,on chain] (b) [label=right: B] {\clique};
  \node[cnode,on chain] (c) [label=right: C] {\anticlique};
\end{scope}

\begin{scope}[every edge/.style={draw=black,very thick},
    every loop/.style={}]
    \path [-] (a) edge (b);
    \path [-] (u1) edge (a);
    \path [-] (u3) edge (a);
    \path [-] (u5) edge (a);
    \path [-] (u4) edge (b);
    \path [-] (u5) edge (b);
\end{scope}
\end{tikzpicture}
      \end{framed}
      \caption{Example with $X = \lbrace 2, 4\rbrace$ and $Y = \lbrace 1, 2, 3 \rbrace$.\label{fig:example1}}
    \end{subfigure}
    \begin{subfigure}[t]{0.325\textwidth}
      \begin{framed}
          \centering
          \definecolor{myred}{RGB}{217,185,185}
\definecolor{myblue}{RGB}{185,185,217}
\definecolor{mygreen}{RGB}{185,217,185}
\definecolor{mygrey}{RGB}{185,185,185}

\begin{tikzpicture}[thick,
  enode/.style={draw,circle,thick},
  unode/.style={enode,fill=myred},
  anode/.style={enode,fill=myblue},
  bnode/.style={enode,fill=mygreen},
  cnode/.style={enode,fill=mygrey},
]

\begin{scope}[start chain=going below,node distance=5pt]
    \foreach \i in {1,2,...,5}
        \node[unode,on chain] (u\i) [label=left: $U_{\i}$] {\anticlique};
\end{scope}

\begin{scope}[xshift=50pt,yshift=-20pt,start chain=going below,node distance=10pt]
  \node[anode,on chain] (a) [label=right: A] {\clique};
  \node[bnode,on chain] (b) [label=right: B] {\clique};
  \node[cnode,on chain] (c) [label=right: C] {\anticlique};
\end{scope}

\begin{scope}[every edge/.style={draw=black,very thick},
    every loop/.style={}]
    \path [-] (a) edge (b);
    \path [-] (u1) edge (a);
    \path [-] (u3) edge (a);
    \path [-] (u5) edge (a);
    \path [-] (u2) edge (b);
    \path [-] (u4) edge (b);
    \path [-] (u5) edge (b);
\end{scope}
\end{tikzpicture}
      \end{framed}
      \caption{Example with $X = \lbrace 2, 4\rbrace$ and $Y = \lbrace 1, 3 \rbrace$.\label{fig:example2}}
    \end{subfigure}
  }
\end{figure}


\bibliography{cdk17}

\begin{thebibliography}{10}

\bibitem{abgpry16}
Maryam Aliakbarpour, Amartya~Shankha Biswas, Themistoklis Gouleakis, John
  Peebles, Ronitt Rubinfeld, and Anak Yodpinyanee.
\newblock Sublinear-time algorithms for counting star subgraphs with
  applications to join selectivity estimation.
\newblock {\em CoRR}, abs/1601.04233, 2016.

\bibitem{Assadi:Khanna:Li:17}
Sepehr Assadi, Sanjeev Khanna, and Yang Li.
\newblock On estimating maximum matching size in graph streams.
\newblock In {\em {ACM-SIAM} Symposium on Discrete Algorithms}, pages
  1723--1742, 2017.

\bibitem{bc15}
Vladimir Braverman and Stephen~R. Chestnut.
\newblock Universal sketches for the frequency negative moments and other
  decreasing streaming sums.
\newblock In {\em {APPROX/RANDOM}}, pages 591--605, 2015.

\bibitem{Cabello2015}
Sergio Cabello and Pablo P{\'e}rez-Lantero.
\newblock {\em Interval Selection in the Streaming Model}, pages 127--139.
\newblock Springer International Publishing, Cham, 2015.

\bibitem{c79}
Y.~Caro.
\newblock New results on the independence number.
\newblock Technical report, Tel Aviv University, 1979.

\bibitem{Cormode:Jowhari:Monemizadeh:Muthukrishnan:16}
Graham Cormode, Hossein Jowhari, Morteza Monemizadeh, and S.~Muthukrishnan.
\newblock The sparse awakens: Streaming algorithms for matching size estimation
  in sparse graphs.
\newblock {\em CoRR}, abs/1608.03118, 2016.

\bibitem{ers16}
Talya Eden, Dana Ron, and C.~Seshadhri.
\newblock Sublinear time estimation of degree distribution moments: The
  arboricity connection.
\newblock {\em CoRR}, abs/1604.03661, 2016.

\bibitem{ehr16}
Yuval Emek, Magn{\'{u}}s~M. Halld{\'{o}}rsson, and Adi Ros{\'{e}}n.
\newblock Space-constrained interval selection.
\newblock {\em {ACM} Trans. Algorithms}, 12(4):51, 2016.

\bibitem{Feigenbaum:Kannan:McGregor:Suri:Zhang:04}
J.~Feigenbaum, S.~Kannan, A.~McGregor, S.~Suri, and J.~Zhang.
\newblock On graph problems in a semi-streaming model.
\newblock In {\em Proceedings of the International Colloquium on Automata,
  Languages, and Programming}, 2004.

\bibitem{grs11}
Mira Gonen, Dana Ron, and Yuval Shavitt.
\newblock Counting stars and other small subgraphs in sublinear-time.
\newblock {\em {SIAM} J. Discrete Math.}, 25(3):1365--1411, 2011.

\bibitem{griggs83}
Jerrold~R Griggs.
\newblock Lower bounds on the independence number in terms of the degrees.
\newblock {\em Journal of Combinatorial Theory, Series B}, 34(1):22 -- 39,
  1983.

\bibitem{hhls16}
Bjarni~V. Halld{\'o}rsson, Magn{\'u}s~M. Halld{\'o}rsson, Elena Losievskaja,
  and Mario Szegedy.
\newblock Streaming algorithms for independent sets in sparse hypergraphs.
\newblock {\em Algorithmica}, 76(2):490--501, 2016.

\bibitem{hk15}
Magn{\'{u}}s~M. Halld{\'{o}}rsson and Christian Konrad.
\newblock Distributed large independent sets in one round on
  bounded-independence graphs.
\newblock In {\em Distributed Computing}, pages 559--572, 2015.

\bibitem{hssw12}
Magn\'{u}s~M. Halld\'{o}rsson, Xiaoming Sun, Mario Szegedy, and Chengu Wang.
\newblock Streaming and communication complexity of clique approximation.
\newblock In {\em International Colloquium on Automata, Languages, and
  Programming}, pages 449--460, 2012.

\bibitem{Hastad99}
Johan H{\aa}stad.
\newblock Clique is hard to approximate within $n^{1-\epsilon}$.
\newblock {\em Acta Mathematica}, 182(1):105--142, 1999.

\bibitem{Henzinger:Raghavan:Rajagopalan:98}
M.~Henzinger, P.~Raghavan, and S.~Rajagopalan.
\newblock Computing on data streams.
\newblock Technical Report SRC 1998-011, DEC Systems Research Centre, 1998.

\bibitem{ks92}
B.~Kalyanasundaram and G.~Schnitger.
\newblock The probabilistic communication complexity of set intersection.
\newblock {\em SIAM Journal on Discrete Mathematics}, 5(4):545--557, 1992.

\bibitem{k72}
R.~M. Karp.
\newblock {Reducibility among combinatorial problems}.
\newblock In R.~E. Miller and J.~W. Thatcher, editors, {\em Complexity of
  Computer Computations}, pages 85--103. Plenum Press, 1972.

\bibitem{McGregor:14}
Andrew McGregor.
\newblock Graph stream algorithms: a survey.
\newblock {\em {SIGMOD} Record}, 43(1):9--20, 2014.

\bibitem{Paz:Schwartzman:17}
Ami Paz and Gregory Schwartzman.
\newblock A ($2 + \epsilon$)-approximation for maximum weight matching in the
  semi-streaming model.
\newblock In {\em {ACM-SIAM} Symposium on Discrete Algorithms}, pages
  2153--2161, 2017.

\bibitem{t48}
Paul Tur{\'a}n.
\newblock On an extremal problem in graph theory.
\newblock {\em Mat. Fiz. Lapok}, 48(436-452):137, 1941.

\bibitem{w81}
V.K. Wei.
\newblock A lower bound on the stability number of a simple graph.
\newblock Technical report, Bell Laboratories, 1981.

\bibitem{Zuckerman07}
David Zuckerman.
\newblock Linear degree extractors and the inapproximability of max clique and
  chromatic number.
\newblock {\em Theory of Computing}, 3(1):103--128, 2007.

\end{thebibliography}


\end{document}